\documentclass[a4paper,UKenglish]{lipics-cust}

\usepackage{hyperref}
\usepackage{graphicx}
\usepackage{float}
\usepackage{subfig}
\usepackage{footnote}
\usepackage[classfont=slant]{complexity}
\usepackage{hhline}
\usepackage{multirow}
\usepackage{rotating}
\usepackage{tikz}
\usepackage{relsize}
\usepackage{ifthen}
\usepackage[toc,page]{appendix}
\usepackage{amsfonts}
\usepackage{hyphenat}
\usepackage{microtype} 
\usepackage{xspace}

\begin{document}

\title{Parameterized Complexity of Graph Constraint Logic}
\titlerunning{Parameterized Complexity of Graph Constraint Logic}

\author{Tom C. van der Zanden}

\affil{Department of Computer Science, Utrecht University, Utrecht, The Netherlands \\\texttt{T.C.vanderZanden@uu.nl}}
\authorrunning{T.\,C. van der Zanden}

\Copyright{Tom C. van der Zanden}



\maketitle

\renewcommand{\theenumi}{\alph{enumi}}

\usetikzlibrary{arrows,automata,decorations.markings,calc,positioning}

\makeatletter
\tikzset{
    position/.style args={#1 degrees from #2}{
        at=(#2.#1), anchor=#1+180, shift=(#1:\tikz@node@distance)
    }
}
\makeatother

\tikzset{fontscale/.style = {font=\relsize{#1}}}

\tikzset{->-/.style={decoration={
  markings,
  mark=at position #1 with {\arrow{>}}},postaction={decorate}}}

\tikzset{-<-/.style={decoration={
  markings,
  mark=at position #1 with {\arrow{<}}},postaction={decorate}}}

\tikzset{
    between/.style args={#1 and #2}{
         at = ($(#1)!0.5!(#2)$)
    }
}

\DeclareGraphicsExtensions{.pdf}
\graphicspath{{./img/}}

\renewcommand\theContinuedFloat{\alph{ContinuedFloat}}

\newlength{\problemoffset}
\setlength{\problemoffset}{0.0cm}

\newcommand{\decision}[3]{
\begin{list}{}{
\setlength{\leftmargin}{\problemoffset}
\setlength{\rightmargin}{\problemoffset}
\setlength{\parsep}{0pt}
\setlength{\itemsep}{2pt}
\setlength{\topsep}{\itemsep}
\setlength{\partopsep}{\itemsep}
}
\item
{\textsc{\normalsize #1}}
\item
{\textbf{Instance:} #2}
\item
{\textbf{Question:} #3}
\end{list}
}

\newcommand{\ppdecision}[4]{
\begin{samepage} 
\begin{list}{}{
\setlength{\leftmargin}{\problemoffset}
\setlength{\rightmargin}{\problemoffset}
\setlength{\parsep}{0pt}
\setlength{\itemsep}{2pt}
\setlength{\topsep}{\itemsep}
\setlength{\partopsep}{\itemsep}
}
\item
{\textsc{\normalsize #1}}
\item
{\textbf{Instance:} #2}
\item
{\textbf{Parameter:} #3}
\item
{\textbf{Question:} #4}
\end{list}
\end{samepage}
}

\newclass{\EXPTIME}{EXPTIME}
\newclass{\NEXPTIME}{NEXPTIME}

\newcommand{\NCL}[1][]{
\ifthenelse{\equal{#1}{}}
	{\textsc{\normalsize Nondeterministic Constraint Logic}\xspace}
	{\textsc{\normalsize #1 Nondeterministic Constraint Logic}\xspace}
}

\newcommand{\scNCL}[1][]{
\ifthenelse{\equal{#1}{}}
	{\textsc{\normalsize Nondeterministic Constraint Logic}\xspace}
	{\textsc{\normalsize #1 Nondeterministic Constraint Logic}\xspace}
}

\newcommand{\CGS}[0]{{\textsc{\normalsize Constraint Graph Satisfiability}}\xspace}

\newcommand{\scCGS}[1]{{\textsc{\normalsize Constraint Graph Satisfiability}}\xspace}

\newtheorem*{claim}{\theoremstyle{remark}Claim}

\begin{abstract}
Graph constraint logic is a framework introduced by Hearn and Demaine \cite{hearn02}, which provides several problems that are often a convenient starting point for reductions. We study the parameterized complexity of \textsc{\normalsize Constraint Graph Satisfiability} and both bounded and unbounded versions of \NCL (NCL) with respect to solution length, treewidth and maximum degree of the underlying constraint graph as parameters. As a main result we show that restricted NCL remains $\PSPACE$-complete on graphs of bounded bandwidth, strengthening Hearn and Demaine's framework. This allows us to improve upon existing results obtained by reduction from NCL. We show that reconfiguration versions of several classical graph problems (including independent set, feedback vertex set and dominating set) are $\PSPACE$-complete on planar graphs of bounded bandwidth and that Rush Hour, generalized to $k\times n$ boards, is $\PSPACE$-complete even when $k$ is at most a constant.
\end{abstract}

\section{Introduction}

\NCL (NCL) was introduced by Hearn and Demaine in \cite{hearn02} and extended in \cite{hdbook} to a more general graph constraint logic framework. The framework provides a number of problems complete for various complexity classes, that aim to provide a convenient starting point for reductions proving the hardness of games and puzzles. We study the \textsc{\normalsize Constraint Graph Satisfiability} problem and \NCL in a parameterized setting, considering (combinations of) solution length, treewidth and maximum degree of the underlying constraint graph as parameters.

As part of the constraint logic framework \cite{hearn02}, Hearn and Demaine provide a restricted variant of \NCL (\textsc{\normalsize restricted NCL}), in which the constraint graph is planar, 3-regular, uses only weights in $\{1,2\}$ and the graph is constructed from only two specific vertex types (AND and OR).  \textsc{\normalsize Restricted NCL} is $\PSPACE$-complete, and is (due to the restrictions) a particularly suitable starting point for reductions. Hearn and Demaine's reduction creates graphs of unbounded treewidth. We strengthen their result, by providing a new reduction showing that \textsc{\normalsize restricted NCL} remains $\PSPACE$-complete, even when restricted to graphs of bandwidth at most a given constant (which is a subclass of graphs of treewidth at most a given constant). We show hardness by reduction from \textsc{\normalsize $H$-Word Reconfiguration} \cite{wrochna14}.

The puzzle game Rush Hour, when generalized to $n\times n$ boards, is $\PSPACE$-complete \cite{bf02}. Hearn and Demaine provide a reduction from NCL to Rush Hour \cite{hearn05}. As a consequence of this reduction and our improved hardness result for NCL, we show that Rush Hour is $\PSPACE$-complete even when played on $k\times n$ boards, where $k$ is a constant. This is in contrast to the result of Ravikumar \cite{pegs} that Peg Solitaire, a game $\NP$-complete on $n\times n$ boards, is linear time solvable on $k\times n$ boards for any fixed $k$.

NCL is also has applications in showing the hardness of reconfiguration problems \cite{is-r,ds-r,gcol,ito08}. For some reconfiguration problems, their hardness on planar graphs of low maximum degree is known by reduction from NCL \cite{ds-r,ito08} while their hardness on bounded bandwidth graphs is known by reduction from \textsc{\normalsize $H$-Word Reconfiguration} \cite{wrochna14,ds-r,raman14}. Our reduction, which combines techniques from Hearn and Demaine's constraint logic \cite{hdbook} and the reductions from \textsc{\normalsize $H$-Word Reconfiguration} in \cite{wrochna14,raman14} unifies these results: we show that a number of reconfiguration problems (including Independent Set, Vertex Cover and Dominating Set) are $\PSPACE$-complete on low-degree, planar graphs of bounded bandwidth. Previously, hardness was known only on graphs that \emph{either} are planar and have low degree \emph{or} have bounded bandwidth - we show that the problems remain hard even when both of these conditions hold simultaneously. Note that while a graph of bounded bandwidth also has bounded degree, the graphs created in these reductions have quite large bandwidth. The degree bounds we obtain are much tighter than what would be obtained from the bandwidth bound alone.

Our results concerning the hardness of constraint logic problems are summarized in Table \ref{tab:overview}. This table shows the parameterized complexity of \textsc{\normalsize Constraint Graph Satisfiability} (CGS), unbounded configuration-to-edge (C2E) and configuration-to-configuration (C2C) variants of \NCL and their respective bounded counterparts (\textsc{\normalsize bC2E} and \textsc{\normalsize bC2C}) with respect to solution length ($l$), maximum degree ($\Delta$) and treewidth ($tw $). If a traditional complexity class is listed this means that the problem is hard for this class even when restricted to instances where the parameter is at most a constant.

\begin{table}[h]
  \centering
  \caption{Parameterized complexity of graph constraint logic problems.}
  \label{tab:overview}
  \begin{tabular}{|c|c|c|c|c|c|c|c|}
    \cline{3-7}
    \multicolumn{2}{c|}{} & \multicolumn{5}{c|}{Problem} \\
    \cline{3-7}
    \multicolumn{2}{c|}{} & \textsc{CGS} & \textsc{C2E} & \textsc{C2C} & \textsc{bC2E} & \textsc{bC2C} \\
    \hline
    \multirow{3}{*}{\begin{sideways}  Parameters ~~ \end{sideways}} & - & $\NP$-C & $\PSPACE$-C & $\PSPACE$-C & $\NP$-C & $\NP$-C \\
    \cline{2-7}
    & $l$ & - & $\W[1]$-hard & $\W[1]$-hard & $W[1]$-hard & $\FPT$ \\
    \cline{2-7}
    & $l+\Delta$ & - & $\FPT$ & $\FPT$ & $\FPT$ & $\FPT$ \\
    \cline{2-7}
    & $tw$ & weakly $\NP$-C & $\PSPACE$-C & $\PSPACE$-C & weakly $\NP$-H & weakly $\NP$-H \\
    \cline{2-7}
    & $tw+\Delta$ & $\FPT$ & $\PSPACE$-C & $\PSPACE$-C & $\FPT$ & $\FPT$ \\
    \hline
  \end{tabular}
 
\end{table}

\section{Preliminaries}

\subsection{Constraint Logic}

\begin{definition}[Constraint Graph]
A \emph{constraint graph} is a graph with edge weights and vertex weights. A \emph{legal configuration} for a constraint graph is an assignment of an orientation to each edge such that for each vertex, the total weight of the edges pointing into that vertex is at least that vertex' weight (its \emph{minimum inflow}).
\end{definition}

\begin{figure}[b]
     \centering
     \hfill
     \subfloat[][OR vertex] {
         \begin{tikzpicture}[>=stealth',shorten >=1pt,auto,node distance=1.5cm,line width=2.5pt]
         \tikzstyle{every state}=[fill=none,draw=none,thin,text=black,scale=0.6]

  \node[state,fill=gray!60,scale=1.3,draw=black]         (A)                             {\textbf{2}};
  \node[state]         (B) [position=330 degrees from A] {};
  \node[state]         (D) [position=210 degrees from A] {};
  \node[state]         (C) [position=90 degrees from A] {};

  \path (A) edge[blue]              node[black] {2} (B)
                 edge[blue]              node[black] {2} (C)
           (D)  edge[blue]              node[black] {2} (A);
\end{tikzpicture} 
         \label{fig:orv}
     }
     \hfill
     \subfloat[][AND vertex] {
         \begin{tikzpicture}[>=stealth',shorten >=1pt,auto,node distance=1.5cm,line width=2.5pt]
         \tikzstyle{every state}=[fill=none,draw=none,thin,text=black,scale=0.6]

  \node[state,fill=gray!60,scale=1.3,draw=black]         (A)                             {\textbf{2}};
  \node[state]         (B) [position=330 degrees from A] {};
  \node[state]         (D) [position=210 degrees from A] {};
  \node[state]         (C) [position=90 degrees from A] {};

  \path (A) edge[red,semithick]              node[black] {1} (B)
                 edge[blue]              node[black] {2} (C)
           (D)  edge[red,semithick]              node[black] {1} (A);
\end{tikzpicture} 
         \label{fig:andv}
     }
     \hfill\null
     \caption{The two vertex types from which a restricted constraint graph is constructed: (a) OR vertex and (b) AND vertex. Following the convention set in \cite{hearn02}, as a mnemonic weight $2$ edges are drawn blue (dark grey) and thick, while weight $1$ edges are drawn red (light grey) and thinner.}
     \label{fig:andorv}
\end{figure}
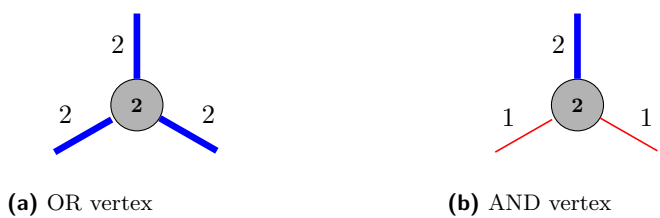

\newpage

A fundamental decision problem regarding constraint graphs is that of their satisfiability:

\decision{Constraint Graph Satisfiablility}{A constraint graph $G$.}{Does $G$ have a legal configuration?}

\textsc{\normalsize Constraint Graph Satisfiability (CGS)} is $\NP$-complete \cite{hdbook}.

An important problem regarding constraint graph configurations is whether they can be reconfigured into each other:

\decision{Nondeterministic Constraint Logic (C2C)}{A constraint graph $G$ and two legal configurations $C_1,C_2$ for $G$.}{Is there a sequence of legal configurations from $C_1$ to $C_2$, where every configuration is obtained from the previous configuration by changing the orientation of one edge?}

This problem is called the configuration-to-configuration (C2C) variant of \NCL. It is $\PSPACE$-complete \cite{hdbook}. The configuration\hyp{}to\hyp{}edge variant (C2E) is also $\PSPACE$-complete \cite{hdbook}:

\decision{Nondeterministic Constraint Logic (C2E)}{A constraint graph $G$, a target edge $e$ from $G$ and an initial legal configuration $C_1$ for $G$.}{Is there a sequence of legal configurations, starting with $C_1$, where every configuration is obtained from the previous by changing the orientation of one edge, so that eventually $e$ is reversed?}

For the \textsc{\normalsize C2C} and \textsc{\normalsize C2E} problems, \textsc{\normalsize bC2C} and \textsc{\normalsize bC2E} denote their bounded variants which ask whether there exists a reconfiguration sequence in which each edge is reversed at most once. These problems are $\NP$-complete \cite{hdbook}.

Hearn and Demaine \cite{hdbook} consider a restricted subset of constraint graphs, which are planar and constructed using only two specific types of vertices: AND and OR vertices (Figure \ref{fig:andorv}).

The OR vertex has minimum inflow $2$ and three incident weight $2$ edges. Its inflow constraint is thus satisfied if and only if at least one of its incident edges is directed inwards (resembling an OR logic gate). The AND vertex has minimum inflow $2$, two incident weight $1$ edges and one incident weight $2$ edge. Its constraint is thus satisfied if and only if both weight $1$ edges are directed inwards or the weight $2$ edge is directed inwards (resembling an AND logic gate). Both C2C and C2E NCL remain $\PSPACE$-complete under these restrictions.

\subsection{Constraint Logic Gadgets}

Some constructions in this paper use existing gadgets (such as crossover) for NCL reductions due to Hearn and Demaine \cite{hdbook}. For the purpose of being self-contained, we reproduce these gadgets here and state (without proof) their functionality.

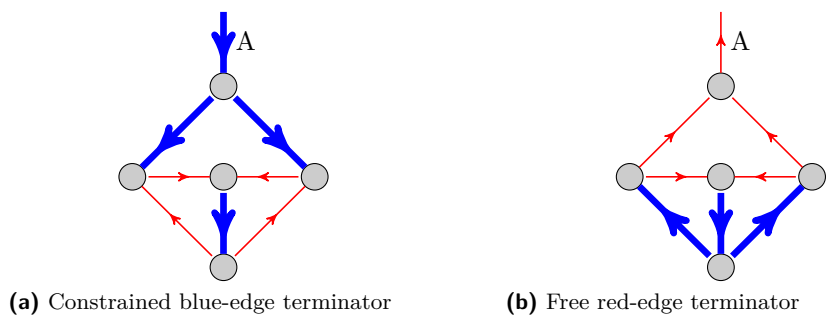
\begin{figure}[b]
     \centering
     \hfill
     \subfloat[][Constrained blue-edge terminator] {
     \hspace{1.2cm} \begin{tikzpicture}[>=stealth',shorten >=1pt,auto,node distance=3cm,line width=2.5pt]
         \tikzstyle{every state}=[fill=gray!40,draw=black,thin,text=black,scale=0.4]

  \node[state](A) [] {};
  \node[state](Ae) [above of=A,fill=none,draw=none] {};
  
  \node[state](M) [below of=A] {};         
  \node[state](L) [left of=M] {};         
  \node[state](R) [right of=M] {};         
  \node[state](D) [below of=M] {};

  \path (A) edge[blue,-<-=0.65] node[black,right] {A} (Ae);
  
  \path (L) edge[blue,-<-=0.5] (A)
        (L) edge[red,semithick,-<-=0.5] (D)
        (R) edge[red,semithick,-<-=0.5] (D)
        (R) edge[blue,-<-=0.5] (A)
        (M) edge[red,semithick,-<-=0.5] (L)
        (M) edge[red,semithick,-<-=0.5] (R)
        (D) edge[blue,-<-=0.7] (M);
\end{tikzpicture} \hspace{1.2cm}
         \label{fig:blue-term}
     }
     \hfill
     \subfloat[][Free red-edge terminator] {
\hspace{1.2cm} \begin{tikzpicture}[>=stealth',shorten >=1pt,auto,node distance=3cm,line width=2.5pt]
         \tikzstyle{every state}=[fill=gray!40,draw=black,thin,text=black,scale=0.4]

  \node[state](A) [] {};
  \node[state](Ae) [above of=A,fill=none,draw=none] {};
  
  \node[state](M) [below of=A] {};         
  \node[state](L) [left of=M] {};         
  \node[state](R) [right of=M] {};         
  \node[state](D) [below of=M] {};

  \path (A) edge[red,semithick,->-=0.6] node[black,right] {A} (Ae);
  
  \path (L) edge[red,semithick,->-=0.5] (A)
        (L) edge[blue,-<-=0.6] (D)
        (R) edge[blue,-<-=0.6] (D)
        (R) edge[red,semithick,->-=0.5] (A)
        (M) edge[red,semithick,-<-=0.5] (L)
        (M) edge[red,semithick,-<-=0.5] (R)
        (D) edge[blue,-<-=0.7] (M);
\end{tikzpicture} \hspace{1.2cm}
         \label{fig:red-term}
     }
     \hfill\null
     \caption{Gadgets for terminating loose edges. The constrained blue-edge terminator (a) forces the blue edge $A$ to point into the gadget, while the free red-edge terminator (b) allows the red edge $A$ to point out of the gadget.}
     \label{fig:terminators}
\end{figure}

\subparagraph*{Edge Terminators.} The constrained blue-edge (Figure \ref{fig:blue-term}) terminator allows us to have a loose blue edge that is forced to point outwards, effectively removing the edge from the graph while still meeting the requirement that the graph is built from only AND and OR vertices. The free red-edge terminator (Figure \ref{fig:red-term}) allows us to have a loose red edge whose orientation can be freely chosen, effectively decreasing the minimum inflow of the vertex to which it is incident by one.

\subparagraph*{Red-blue Conversion.} It is useful to be able to convert a blue edge to a red edge, i.e. we require a gadget which has a blue edge that can (be reconfigured to) point outwards if and only if its red edge is pointing inwards and vice-versa. Hearn and Demaine \cite{hdbook} provide a construction that allows red-blue conversion in pairs, but also note a simpler construction is possible: an AND vertex, with one of its red edges attached to a free red-edge terminator (Figure \ref{fig:red-term}) can serve as a red-blue conversion gadget.

\subparagraph*{Crossover Gadget.} The crossover gadget (Figure \ref{fig:crossover}) has 4 incident blue edges $A,B,C$ and $D$, with the property that $A$ can (be reconfigured to) point outward only if $B$ is pointing inwards and vice-versa, with the same property also holding for $C$ and $D$.

The crossover gadget is constructed using degree-4 vertices (of minimum inflow 2) that are incident to four weight-1 edges. These can be replaced with the half-crossover gadget (Figure \ref{fig:half-crossover}) to obtain a construction that only uses AND and OR vertices. Note that this replacement requires using red-blue conversion gadgets.

Note that to cross a red edge with a blue edge (or a red edge with another red edge) we can use the aforementioned crossover gadget, paired with red-blue conversion gadgets.

\subparagraph*{Latch Gadget.} The final gadget we require is the latch gadget (Figure \ref{fig:latch}). It can be unlocked by reversing the edge $L$, after which its state (the orientation of edge $A$, and adjacent edges) can be changed. The latch can be locked by reversing the edge $L$ again. It is useful to reduce C2E NCL problems to C2C NCL problems, replacing the target edge in the C2E problem by a latch gadget, and creating a C2C problem in which the goal configuration differs from the start configuration only in the state of the latch gadget.

\begin{figure}[t]
     \centering
     \hfill
     \subfloat[][Crossover] {
         \begin{tikzpicture}[>=stealth',shorten >=1pt,auto,node distance=2cm,line width=2.5pt]
         \tikzstyle{every state}=[fill=gray!40,draw=black,thin,text=black,scale=0.4]

  \node[state](A) [] {};
  \node[state](Ae) [above of=A,fill=none,draw=none] {};
  
  \node[state](AD) [right=0.8 of A] {};  
  \node[state](AC) [left=0.8 of A] {};  
  
  \node[state](C) [below left of=AC] {};
  \node[state](Ce) [left of=C,fill=none,draw=none] {};
  
  \node[state](D) [below right of=AD] {};
  \node[state](De) [right of=D,fill=none,draw=none] {};  

  \node[state](BD) [below left of=D] {};  
  \node[state](BC) [below right of=C] {};  
  
  \node[state](B) [left=0.8 of BD] {};
  \node[state](Be) [below of=B,fill=none,draw=none] {};
  
  \node[state](L) [below right of=AC] {};
  \node[state](R) [below left of=AD] {};  

  \path (A) edge[blue]              node[black,right] {A} (Ae)
           (B) edge[blue]              node[black,right] {B} (Be)
           (C) edge[blue]              node[black,above] {C} (Ce)
           (D) edge[blue]              node[black,above] {D} (De) ;
           
   \path (L) edge[blue] (R);
   
   \path (A) edge[red,semithick] (AD)
            (AD) edge[red,semithick] (D)
            (D) edge[red,semithick] (BD)
            (BD) edge[red,semithick] (B)
            (B) edge[red,semithick] (BC)
            (BC) edge[red,semithick] (C)
            (C) edge[red,semithick] (AC)
            (AC) edge[red,semithick] (A);

   \path (AC) edge[red,semithick] (L)            
            (L) edge[red,semithick] (BC)            
            (BC) edge[red,semithick] (AC) ;         
    \path (AD) edge[red,semithick] (R)            
            (R) edge[red,semithick] (BD)            
            (BD) edge[red,semithick] (AD) ;        
   
\end{tikzpicture} 
         \label{fig:crossover}
     }
     \hfill
     \subfloat[][Half-crossover] {
\begin{tikzpicture}[>=stealth',shorten >=1pt,auto,node distance=2cm,line width=2.5pt]
         \tikzstyle{every state}=[fill=gray!40,draw=black,thin,text=black,scale=0.4]

  \node[state](A) [] {};
  \node[state](Ae) [above of=A,fill=none,draw=none] {};
  
  \node[state](AD) [right=0.6 of A] {};  
  \node[state](AC) [left=0.6 of A] {};  
  
  \node[state](C) [below left of=AC] {};
  \node[state](Ce) [left of=C,fill=none,draw=none] {};
  
  \node[state](D) [below right of=AD] {};
  \node[state](De) [right of=D,fill=none,draw=none] {};  

  \node[state](BD) [below left of=D] {};  
  \node[state](BC) [below right of=C] {};  
  
  \node[state](B) [left=0.6 of BD] {};
  \node[state](Be) [below of=B,fill=none,draw=none] {};
  
  \node[state](L) [below right of=AC,fill=none,draw=none] {};
  \node[state](R) [below left of=AD,fill=none,draw=none] {};  

  \path (A) edge[blue]              node[black,right] {} (Ae)
           (B) edge[blue]              node[black,right] {} (Be)
           (C) edge[blue]              node[black,above] {} (Ce)
           (D) edge[blue]              node[black,above] {} (De) ;
           
   
   \path (A) edge[blue] (AD)
            (AD) edge[red,semithick] (D)
            (D) edge[red,semithick] (BD)
            (BD) edge[blue] (B)
            (B) edge[blue] (BC)
            (BC) edge[red,semithick] (C)
            (C) edge[red,semithick] (AC)
            (AC) edge[blue] (A);

   \path (BC) edge[red,semithick] (AC) ;         
    \path (BD) edge[red,semithick] (AD) ;        
   
\end{tikzpicture} 
         \label{fig:half-crossover}
     }
     \hfill\null
     \caption{The crossover gadget (a) may be constructed with help of the half-crossover gadget (b).}
     \label{fig:crossovers}
\end{figure}

\begin{figure}[t]
     \centering
     \begin{tikzpicture}[>=stealth',shorten >=1pt,auto,node distance=3cm,line width=2.5pt]
         \tikzstyle{every state}=[fill=gray!40,draw=black,thin,text=black,scale=0.4]

  \node[state](A) [] {};
  \node[state](Ae) [left of=A,fill=none,draw=none] {};
  
  \node[state](T) [position=30 degrees from A] {};         
  \node[state](B) [position=330 degrees from A] {};         
  \node[state](Te) [right of=T,fill=none,draw=none] {};         
  \node[state](Be) [right of=B,fill=none,draw=none] {};         
  
  \path (A) edge[blue,->-=0.65] node[black,above] {L} (Ae);
  
  \path (Te) edge[red,semithick,-<-=0.65] (T)
           (T) edge[red,semithick,->-=0.5] node[black,right] {A} (B)
           (B) edge[red,semithick,-<-=0.5] (Be)
           (A) edge[blue,->-=0.6] (T)
           (A) edge[blue,-<-=0.6] (B) ;
\end{tikzpicture}
     \caption{The latch gadget can be (un-)locked by reversing edge $L$, and edge $A$ may reverse if and only if the gadget is unlocked. The gadget is shown in its locked state.}
     \label{fig:latch}
\end{figure}

\subsection{$H$-Word Reconfiguration}

To show hardness of NCL on bounded bandwidth graphs, we reduce from the \textsc{\normalsize $H$-Word Reconfiguration} problem, introduced in \cite{wrochna14}.

\begin{definition}[$H$-word]
Let $H=(\Sigma,E)$ where $\Sigma$ is an alphabet and $E\subseteq \Sigma \times \Sigma$ a relation. An \emph{$H$-word} is a word over $\Sigma$ such that every pair of consecutive characters $(a,b)$ is an element of $E$.
\end{definition}

\decision{$H$-Word Reconfiguration}{Two $H$-words $W_s,W_g$ of equal length}{Is there a sequence of $H$-words $W_1,\ldots,W_n$, so that every pair of consecutive words $W_i,W_{i+1}$ can be obtained from each other by changing one character to another and $W_1=W_s,W_n=W_g$?}

\begin{theorem}[Wrochna \cite{wrochna14}]\label{thm:hwordhard}
There exists an $H$ such that \textsc{\normalsize $H$-Word Reconfiguration} is $\PSPACE$-complete.
\end{theorem}

\subsection{Graph Parameters}

As one possible parametrization we consider treewidth. The main result (Theorems \ref{thm:c2cnclhard} and \ref{thm:c2enclhard}) considers graphs of bounded bandwidth, a subclass of graphs of bounded treewidth \cite{bodl98}.

\begin{definition}[Treewidth]
A \emph{tree decomposition} of a graph $G=(V,E)$, is a tree $T$ in which every vertex $t$ of $T$ is associated with a bag $X_t\subseteq V$, such that:
\begin{itemize}
    \item For all $v\in V$, there is a $t$ such that $v\in X_t$
    \item For all $(u,v)\in E$, there is a $t$ such that $u\in X_t$ and $v\in X_t$
    \item For all $v\in V$, $T[\{t\in T:v\in X_t\}]$ is connected
\end{itemize}
The \emph{width} of a tree decomposition is $\max_{t\in T} |X_t| - 1$. The \emph{treewidth} of $G$ is the minimum width over all tree decompositions of $G$.
\end{definition}

The main result concerns graphs of bounded bandwidth. The proof of Theorem \ref{thm:rushhour} uses the notion of cutwidth.

\begin{definition}[Bandwidth \cite{bodl98}]
Let $G=(V,E)$ be a graph and define a one-to-one correspondence $f:V\to \{1,\ldots,|V|\}$. The \emph{bandwidth} of a graph is the minimum over all such correspondences of $\max_{(u,v)\in E} |f(u)-f(v)|$.
\end{definition}

\begin{definition}[Cutwidth \cite{bodl98}]
Let $G=(V,E)$ be a graph and define a one-to-one correspondence $f:V\to \{1,\ldots,|V|\}$. The \emph{cutwidth} of a graph is the minimum over all such correspondences of $\max_{w\in V} |\{(u,v)\in E|f(u)\leq f(w) < f(v)\}|$.
\end{definition}

\subsection{Dynamic Programming on Tree Decompositions}

For some of the positive results in this paper, we use the technique of dynamic programming on tree decompositions. For an excellent survey of and introduction to this topic, see \cite{bodl97}.

To simplify the algorithms, we assume tree decompositions are given in nice form, that is, all of the nodes of the tree decomposition are of one of the following types:

\begin{itemize}
\item \textbf{Leaf}: all leaves of the tree decomposition contain exactly one vertex.

\item \textbf{Introduce}: a node with one child that contains the same set of vertices as its child, with the addition of one new vertex.

\item \textbf{Forget}: a node with one child that contains the same set of vertices as its child, but with one vertex removed.

\item \textbf{Join}: a node with two children so that both the node itself and both of its children contain the exact same set of vertices.
\end{itemize}

Given a tree decomposition, it is possible to find a nice tree decomposition of the same width (and of linear size) in linear time.

\section{Parametrization by Treewidth}

\subsection{Constraint Graph Satisfiability}
We first examine \textsc{Constraint Graph Satisfiability (CGS)} with respect to treewidth $tw$ and maximum degree $\Delta$ as parameters. CGS is strongly $\NP$-complete even for graphs where $\Delta = 3$ \cite{hdbook}. We show that CGS is weakly $\NP$-complete for graphs of treewidth at most $2$, but that CGS is fixed parameter tractable with respect to $tw+\Delta$ and solvable in polynomial time for fixed $tw$ if the input is given in unary. Note that even though CGS becomes fixed parameter tractable with respect to $tw+\Delta$, its reconfiguration variant remains hard (Theorem \ref{thm:c2cnclhard}) even when this parameter is bounded by a constant.

\begin{theorem}\label{thm:cgshard}
\textsc{Constraint Graph Satisfiability} is weakly $\NP$-complete on constraint graphs of treewidth $2$.
\end{theorem}

\begin{proof}
By reduction from \textsc{Partition}. Let $x_1,\ldots,x_n\in \mathbb{N}$ such that $\Sigma_{i=1}^n x_i = 2N$ be an instance of \textsc{Partition}. Construct a constraint graph with vertices $U,W$ with minimum inflow $N$, and vertices $v_1\ldots,v_n$ where the minimum inflow of $v_i$ is equal to $x_i$. For $1\leq i\leq n$, create edges $(U,v_i)$, $(W,v_i)$, both of weight $x_i$.
In any legal configuration for this constraint graph, either edge $(U,v_i)$ or $(W,v_i)$ must be directed towards vertex $v_i$ and the other edge will be directed towards either $U$ or $W$. A solution to the partition problem and a legal configuration for the constraint graph correspond as follows: for all $x_i$ in one half of the partition, the edge $(U,v_i)$ is directed towards $U$ (and the edge $(W,v_i)$ is directed towards $v_i$) and for all $x_i$ in the other half of the partition, the edge $(W,v_i)$ is directed towards $W$ (and the edge $(U,v_i)$ is directed towards $v_i$).
The graph restricted to vertices $\{v_1,\ldots,v_n,U\}$ forms a tree, and thus the constraint graph itself has treewidth 2 (we can add the vertex $W$ to all bags). We have thus shown CGS weakly $\NP$-hard on constraint graphs of treewidth $2$. 
\end{proof}

The following theorem shows that CGS can be solved in polynomial time on graphs of treewidth $2$ when the input is given in unary (and completes the proof of theorem \ref{thm:cgshard}).

\begin{theorem}\label{thm:cgseasy}
\textsc{Constraint Graph Satisfiability} is fixed parameter tractable with respect to parameter $tw+\Delta$ and solvable in time $O^*(|x|^{tw})$ on constraint graphs of treewidth $tw$ when the size of the input numbers given in unary is $|x|$.
\end{theorem}

\begin{proof}
We provide a dynamic programming algorithm on tree decompositions. In the case where $\Delta$ is a parameter, in each bag we store for every combination of orientations of the edges incident to vertices in that bag (at most $2^{\Delta (tw + 1)}$ combinations) whether that combination leads to a valid configuration for the subgraph induced by the subtree rooted at that bag. We assume that the tree decomposition is given as a nice tree decomposition:

\begin{itemize}
\item \textbf{Leaf}: In a leaf node, we enumerate all possible combinations of orientations for edges incident to this vertex (at most $2^\Delta$), and see which combinations satisfy the minimum inflow of that vertex.

\item \textbf{Introduce}: In an introduce node, enumerate all combinations of orientations of edges incident to the vertices in the bag (at most $2^{\Delta (tw + 1)}$). A combination is valid if it satisfies the inflow of the vertex being introduced, and (when restricted to the appropriate edges) is also a valid assignment of orientations for the child bag.

\item \textbf{Forget}: In a forget node, an assignment of orientations is valid if it can be extended (by picking appropriate orientations for the missing edges) to a valid configuration for the child bag.

\item \textbf{Join}: In a join node, an assignment of orientations is valid if it is valid for both child nodes.
\end{itemize}

The amount of work done in each bag is $O^*(2^{\Delta (tw + 1)})$, so the total running time is $O^*(|I| \cdot 2^{\Delta (tw + 1)})$, where $|I|$ is the number of bags (which can be assumed to be linear in the size of the input graph, since there always a linear size tree decomposition).

A similar algorithm can be used in the case where the input is given in unary. Instead of storing the orientations of the edges directly, we store only the resulting inflow values. For all possible combinations of inflow that the vertices in the bag can be receiving (at most $|x|^{tw+1}$ combinations), we store whether that combination of inflows is attainable using only the edges present in the subgraph induced by the subtree rooted in that bag (i.e. edges to vertices in parent bags are not considered) in configurations that are legal for the vertices in the induced subgraph excluding the vertices in the bag itself. This requires (in particular) modification of the forget case, marking as invalid any inflow configurations that do not satisfy the minimum inflow of the vertex being forgotten (detection of whether a vertex is receiving sufficient inflow is deferred to the forget node instead of being checked in the leaf or introduce cases). 
\end{proof}

\subsection{Unbounded Nondeterministic Constraint Logic}

In this section we prove the main result, namely that restricted \textsc{\normalsize Nondeterministic Constraint Logic} remains $\PSPACE$-complete even when restricted to graphs of bounded bandwidth (which is a subclass of graphs of bounded treewidth).

To show $\PSPACE$-completeness, we reduce from \textsc{\normalsize $H$-Word Reconfiguration}. The bandwidth of the constraint graph created in the reduction will depend only on the size of $H$. Since \textsc{\normalsize $H$-Word Reconfiguration} is $\PSPACE$-complete for a fixed (finite) $H$, we obtain a constant bound on the bandwidth.

\begin{theorem}\label{thm:c2cnclhard}
There exists a constant $c$, such that C2C \textsc{\normalsize Nondeterministic Constraint Logic} is $\PSPACE$-complete on planar constraint graphs of bandwidth at most $c$ that consist of only AND and OR vertices.
\end{theorem}

\begin{proof}
Let $H=(\Sigma,E)$ be so that \textsc{\normalsize $H$-Word Reconfiguration} is $\PSPACE$-complete. We provide a reduction from \textsc{\normalsize $H$-Word Reconfiguration} to (unrestricted) NCL and then show how to adapt this reduction to work for restricted NCL. Let $W_s,W_g$ be an instance of \textsc{\normalsize $H$-Word Reconfiguration} and let $n$ denote the length of $W_s$. In the following, all vertices are given minimum inflow 2.

We create a matrix of vertices $X_{i,j}$ for $i\in\{1,\ldots,n\}, j\in\Sigma$, the \emph{character vertices}. The orientation of edges incident to $X_{i,j}$ will correspond to whether the character at position $i$ in the word is character $j$. For every row $i$ of this matrix we create a \emph{universal vertex} $U_i$ and for every $j\in\Sigma$ we create a blue (weight 2) edge connecting $U_i$ and $X_{i,j}$.

In each row $i\in{1,\ldots,n-1}$ and for all pairs $(A,B)\not\in E$, we create a \emph{relation vertex} $\Delta_{i,A,B}$. We create a red (weight 1) edge connected to $X_{i,A}$ and pass it through a red-blue conversion gadget and connect the blue edge leaving the conversion gadget to $\Delta_{i,A,B}$. We mirror this construction, creating a red edge connected to $X_{i+1,B}$, converting it to blue, and connecting it to $\Delta_{i,A,B}$.

Finally, we connect one additional blue edge to $\Delta_{i,A,B}$ and connect it to a constrained blue-edge terminator. This edge serves no purpose (its inflow can never count towards $\Delta_{i,A,B}$) but to make $\Delta_{i,A,B}$ an OR vertex as claimed.

\begin{figure}[t]
\centering
 \resizebox{!}{7.5cm}{\begin{tikzpicture}[>=stealth',shorten >=1pt,auto,node distance=2.5cm,line width=2.5pt]
         \tikzstyle{every state}=[fill=gray!60,draw=black,thin,text=black,minimum width=4.0em]
         \tikzstyle{every node}=[fontscale=1.25]

  \node[state]         (U) []                  {$U_i$};

  \node[state]         (A) [below right=1.5cm of U] {$X_{i,A}$};
  \node[state]         (B) [right of=A] {$X_{i,B}$};
  \node[state,fill=gray!60]         (C) [right of=B] {$X_{i,C}$};
  \node[state]         (D) [right of=C] {$X_{i,D}$};
  
  \coordinate[above = 1cm of A] (Aa) ;
  \coordinate[above = 1cm of B] (Ba) ;
  \coordinate[above = 1cm of C] (Ca) ;
  \coordinate[above = 1cm of D] (Da) ;
  
  \coordinate[below = 1cm of A] (Ab) ;
  \coordinate[below = 1cm of B] (Bb) ;
  \coordinate[below = 1cm of C] (Cb) ;
  \coordinate[below = 1cm of D] (Db) ;
  
  \node[state]         (A2) [below=4cm of A] {$X_{i+1,A}$};
  \node[state,fill=gray!60]         (B2) [right of=A2] {$X_{i+1,B}$};
  \node[state]         (C2) [right of=B2] {$X_{i+1,C}$};
  \node[state]         (D2) [right of=C2] {$X_{i+1,D}$};
  
  \coordinate[above = 1cm of A2] (A2a) ;
  \coordinate[above = 1cm of B2] (B2a) ;
  \coordinate[above = 1cm of C2] (C2a) ;
  \coordinate[above = 1cm of D2] (D2a) ;
  
  \coordinate[below = 1cm of A2] (A2b) ;
  \coordinate[below = 1cm of B2] (B2b) ;
  \coordinate[below = 1cm of C2] (C2b) ;
  \coordinate[below = 1cm of D2] (D2b) ;  
  
  \node[state]         (U2) [below left=1.5cm of A2]                  {$U_{i+1}$};
  
  \node[state]         (Co1) [between=A and B2] {$\Delta_{i,A,B}$};  
  \coordinate[right = 0.75cm of Co1] (Co1R) ;
  
 
  \path[every edge, blue, bend left,->]
  (U) edge (A)
  (U) edge (B)
  (U) edge[-<-=0.5,-] (C)
  (U) edge (D);
  
  \path[every edge, blue, bend right, ->]
  (U2) edge (A2)
  (U2) edge[-<-=0.5,-] (B2)
  (U2) edge (C2)
  (U2) edge (D2);  
  
  \path[every edge,blue]
   (A) edge[->-=0.6,red,thick] ($(A)!0.5!(Co1)$)
   ($(A)!0.5!(Co1)$) edge[->-=0.75] (Co1)
   (Co1) edge[->-=0.75] ($(Co1)!0.5!(B2)$)
   ($(Co1)!0.5!(B2)$) edge[->-=0.65,red,thick] (B2)
   (Co1) edge[->-=0.6] (Co1R) ;   
   
    
   \path[every edge, red,thick]
   (A) edge[->-=0.5] (Aa)
   (B) edge[->-=0.5] (Ba)
   (C) edge[-<-=0.5] (Ca)
   (D) edge[->-=0.5] (Da);
   
   \path[every edge, red,thick]
   (B) edge[->-=0.5] (Bb)
   (C) edge[-<-=0.5] (Cb)
   (D) edge[->-=0.5] (Db);
   
   \path[every edge, red,thick]
   (A2) edge[->-=0.5] (A2a)
   (C2) edge[->-=0.5] (C2a)
   (D2) edge[->-=0.5] (D2a);
   
   \path[every edge, red,thick]
   (A2) edge[->-=0.5] (A2b)
   (B2) edge[-<-=0.5] (B2b)
   (C2) edge[->-=0.5] (C2b)
   (D2) edge[->-=0.5] (D2b);
  
  \end{tikzpicture}} 
     
     \caption{Slice of two rows from the matrix of vertices (over an alphabet of $\Sigma=\{A,B,C,D\}$), and the construction for enforcing non-adjacency of $A$ and $B$.}
     \label{fig:twred}
\end{figure}
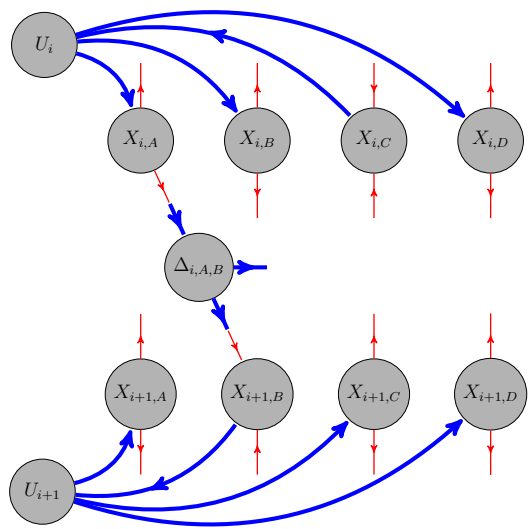

A single instance of this construction is shown in figure \ref{fig:twred}, which depicts a slice of two rows from the matrix of vertices. Having created an instance of this construction for $(A,B)$, we might need to create an instance for $(A,C)$. Rather than attaching another red edge to $X_{i,A}$, we instead split the existing red edge leaving $X_{i,A}$ by connecting it to a red-blue conversion gadget, and attaching the resulting blue edge to an AND vertex. The edges leaving the AND vertex may be oriented outward if and only if the red edge leaving $X_{i,A}$ is oriented into the AND vertex, effectively splitting it. We then convert the two red edges of the AND vertex to blue and attach them to $\Delta_{i,A,B}$ and $\Delta_{i,A,C}$ as before. To create further copies of this construction we can repeat the splitting process, so that $X_{i,A}$ eventually has two incident red edges, one connecting to gadgets in row $i-1$ and one connecting to row $i+1$ (the excess red edges in rows $1$ and $n$ may instead be connected to a free red edge terminator).

We define a character $j$ to be \emph{in} the word at position $i$ if the edge $(U_i,X_{i,j})$ is oriented towards $U_i$.

\begin{claim}
In any legal configuration, at least one character is in the word at each position.
\end{claim}

\begin{proof}
$U_i$ has minimum inflow 2, so at least one of its incident edges $(U_i,X_{i,j})$ must be oriented towards it and hence $j$ is in the word at position $i$. 
\end{proof}

\begin{claim}
In any legal configuration, if $(A,B)\not\in E$ then $A$ can be in the word at position $i$ only if $B$ is not in the word at position $i+1$.
\end{claim}

\begin{proof}
If $A$ is in the word at position $i$, then $X_{i,A}$ is not receiving inflow from $U_i$, so both of the red edges incident to $X_{i,A}$ most point in towards $X_{i,A}$. On the path from $X_{i,A}$ through $\Delta_{i,A,B}$ to $X_{i+1,B}$ we will first pass through a red-blue conversion gadget. Since the red edge is oriented out of the conversion gadget (and towards $X_{i,A}$), the blue edge must be oriented into the conversion gadget. We may then encounter several AND vertices (that are used for the ``splitting'' of red edges incident to $X_{i,A}$), note that since the blue edge is oriented towards the conversion gadget and away from the AND vertex, both red edges must be oriented into the AND vertex. This means that when we encounter another red-blue conversion gadget its blue edge must be oriented into the AND vertex and in turn the red edge incident to the conversion gadget has to be oriented into the conversion gadget. Ultimately, when we arrive at the vertex $\Delta_{i,A,B}$ we find that one of its edges must be oriented out of it (to point into the previously described AND vertices and help satisfy the minimum inflow of $X_{i,A}$).
Suppose by contradiction that $B$ is in the word at position $i+1$. By a similar argument, we find that the second edge incident to $\Delta_{i,A,B}$ also has to be oriented out of it. This is impossible, since the third blue edge incident to $\Delta_{i,A,B}$ is also oriented out of it (since that edge is attached to a constrained blue-edge terminator gadget) and thus its minimum inflow constraint is violated. Therefore $B$ can not be in the word at position $i+1$.
\end{proof}

Note that this claim implies that if for each position we pick \emph{some} character that is in the position at that word, we end up with a valid $H$-Word. We say that a configuration for the constraint graph \emph{encodes} a word $W$ if each of that word's characters is in the word at the appropriate position.

\begin{claim}
If a legal configuration for the constraint graph encoding $W_s$ may be reconfigured into a legal configuration encoding only $W_g$ (i.e. the configuration encodes no other word), then $W_s$ may be reconfigured into $W_g$.
\end{claim}

\begin{proof}
Suppose we have a reconfiguration sequence of legal configurations $C_1,\ldots,C_m$ so that $C_1$ encodes $W_s$ and $C_m$ encodes only $W_g$. We define a reconfiguration sequence of words $W_1,\ldots,W_m$ which has the property that $C_i$ encodes $W_i$. Let $W_1=W_s$, we recursively define $W_i, 1<i\leq m$ as follows: Since $C_{i+1}$ differs from $C_i$ in the orientation of only one edge, the set of words encoded by $C_{i+1}$ differs only from the set of words encoded by $C_i$ by making one character at a position allowed (in the word) or disallowed (was in the word in $C_i$ but not in the word in $C_{i+1}$). If a character at a position becomes allowed in $C_{i+1}$, let $W_{i+1}=W_i$ (since $C_i$ encodes $W_i$, $C_{i+1}$ also encodes $W_i$). If a character at a position becomes disallowed (i.e. is no longer in the word), we obtain $W_{i+1}$ from $W_i$ by changing the character at that position to some allowed character (of which there is at least one). Following these steps, since the final configuration encodes only $W_g$, we obtain a reconfiguration sequence from $W_s$ to $W_g$ as claimed. 
\end{proof}

Note that we may have multiple choices for $W_i$ because $C_i$ can encode more than one word. It suffices to simply pick one of the words it encodes, while ensuring it differs in at most one character from the previous and next.

\begin{claim}
Given a $H$-word $W$ of length $n$, there exists a legal configuration for the constraint graph encoding only $W$.
\end{claim}

\begin{proof}
Pick the orientation of edge $(U_i,X_{i,j})$ to be towards $U_i$ if the character in $W$ at position $i$ is $j$, and towards $X_{i,j}$ otherwise. Clearly this constraint graph encodes only $W$. This configuration can be extended to a legal configuration for the remaining (relation) vertices by noticing the following: if there is a relation vertex $\Delta_{i,A,B}$ then either $A$ is not in the word at position $i$ or $B$ is not in the word at position $j$. Suppose w.l.g. that $B$ is not in the word, then $X_{i+1,B}$ is receiving inflow from $U_{i+1}$ and hence its incident red edges may be pointing outwards, which (after passing through the red-blue conversion gadgets and splitting AND vertices as described before) allows us to satisfy the inflow requirement of $\Delta_{i,A,B}$. 
\end{proof}

\begin{claim}
If two $H$-Words $W_1,W_2$ differ by only one character, then a constraint graph $G$ encoding only $W_1$ can be reconfigured into one encoding only $W_2$.
\end{claim}

\begin{proof}
Suppose that $W_1$ and $W_2$ differ by changing the character at position $i$ from $A$ to $B$. We may first reconfigure $G$ from encoding only $W_1$ to encoding both $W_1$ and $W_2$, and then reconfigure it to encode only $W_2$. This temporarily increases the inflow $U_i$ receives from $2$ (receiving inflow only from $X_{i,A}$) to $4$ (receiving inflow from both $X_{i,A}$ and $X_{i,B}$) back to $2$ (receiving inflow from only $X_{i,A}$). It may be required to reorient some edges to satisfy the inflows of the relation vertices, but this is always possible because neither $A$ nor $B$ conflicts with any preceding or succeeding characters. 
\end{proof}

Note that this claim implies that if $W_s$ can be reconfigured into $W_g$, then a constraint graph encoding only $W_s$ can be reconfigured into one encoding only $W_g$. This completes the reduction.

All that remains to show is that this graph can be adapted so that it only uses AND and OR vertices and becomes planar, and that the resulting graph has bounded bandwidth. The only vertices that are not already AND or OR vertices are the universal vertices $U_i$. Note that taking a single OR vertex (that has 3 incident edges at least one of which has to be oriented inwards) we can attach another OR vertex to one of its edges to obtain a structure that has 4 external edges, at least one of which has to be oriented inwards. Repeating this procedure $n$ times we obtain a tree of OR vertices with $n+3$ external edges, at least one of which has to be oriented inwards, which can replace the universal vertex.

To make the graph planar, we can use Hearn and Demaine's crossover gadget \cite{hdbook}. While this may increase the graph's bandwidth, the following argument that it remains bounded by a constant:

We create a number of bags $B_1,\ldots B_{n-1}$ that are subsets of vertices, with the property that the size of each bag depends only on $H$ (which is fixed) and that edges connect only vertices that are both in the same bag, or a vertex in bag $B_i$ with a vertex in bag $B_{i+1}$. This is achieved by taking in bag $i$ the vertices $\{X_{i,j}:j\in \Sigma\},\{\Delta_{i,A,B}\:(A,B)\not\in E\},\{X_{i+1,j}:j\in \Sigma\}$, the vertices in gadgets between them (i.e. the red-blue conversion gadgets and AND vertices used in the splitting process) and the construction replacing the universal vertices $U_i$ and $U_{i+1}$. This shows the bandwidth of the graph is bounded by a constant $c$, since we can order the vertices so that a vertex in $B_i$ precedes a vertex in $B_{i+1}$ (and pick an arbitrary order for the vertices in the same bag).

We have thus shown how to reduce $H$-Word Reconfiguration to NCL, shown how to adapt our reduction to use only AND and OR vertices and make the resulting graph planar and that the resulting graph has bounded bandwidth and have thus shown Theorem \ref{thm:c2cnclhard}. 
\end{proof}

Theorem \ref{thm:c2cnclhard} also holds for C2E \textsc{\normalsize Nondeterministic Constraint Logic}:

\begin{theorem}\label{thm:c2enclhard}
There is a constant $c$, such that C2E \textsc{\normalsize Nondeterministic Constraint Logic} is $\PSPACE$-complete, even on planar constraint graphs of treewidth (bandwidth) at most $c$ that use only AND and OR vertices.
\end{theorem}

\begin{proof}
\textsc{$H$-Word Reconfiguration} remains $\PSPACE$\hyp{}complete, even when instead of requiring that one $H$-word is reconfigured in to another, we ask whether it is possible to reconfigure a given $H$-word so that a given character appears at a specific position (without requiring anything regarding the remaining characters in the word). This can be seen by examining the proof in \cite{wrochna14}, noting that we may modify the Turing machine to, upon reaching an accepting state, move the head to the start of the tape and write a specific symbol there. The proof of Theorem \ref{thm:c2cnclhard} can then easily be adapted to work for C2E NCL (since a character appearing at a specific position corresponds to reversing a specific edge).
\end{proof}

We note that Theorems \ref{thm:c2cnclhard} and \ref{thm:c2enclhard} may be further strengthened by requiring that all OR vertices in the graph are \emph{protected}, i.e., in any legal configuration at least one of its incident edges is directed outwards. Hearn and Demaine \cite{hdbook} show how to construct an OR vertex using only AND vertices and protected OR vertices.

\subsection{Bounded Nondeterministic Constraint Logic}

We now consider the bounded variant of NCL, in which each edge is allowed to reverse at most once. We show that the bounded variants of NCL are weakly $\NP$-hard on graphs of treewidth at most $2$, but, unlike their unbounded variants, fixed parameter tractable when parameterized by $tw+\Delta$.

\begin{theorem}
Bounded NCL (both C2E and C2C) is weakly $\NP$-hard on constraint graphs of treewidth 2.
\end{theorem}

\begin{proof}
By reduction from \textsc{Partition}. Construct an NCL graph in the same manner as in the proof of Theorem \ref{thm:cgshard}, but create an additional edge $(U,W)$ of weight $N$. In the initial configuration, let the edge $(U,W)$ be directed towards $W$, the edges $(U,v_i)$ towards $U$ and the edges $(W,v_i)$ towards $v_i$.

\begin{itemize}
	\item For the C2E variant, ask whether it is possible to reverse the edge $(U,W)$.
	\item For the C2C variant, ask whether the configuration in which all edges are reversed is reachable.
\end{itemize}

It is possible to reverse $(U,W)$ is and only if there is a partition: initially, vertex $U$ is receiving $2N$ units of inflow (from the edges $(W,v_i)$) and $V$ is receiving $N$ units of inflow (from the edge $(U,W)$). To reverse $(U,W)$, we must reverse some of the edges $(U,v_i)$ and $(W,v_i)$ so that $W$ is receiving exactly $N$ units of inflow from them. If $W$ is receiving less than $N$ units of inflow from these edges it is not possible to reverse $(U,W)$, if it is receiving more than $N$ inflow then $V$ would not be receiving enough inflow. Being able to select these edges so that $W$ receives exactly $N$ units of inflow from them corresponds to a partition.

This completes the proof for the C2E variant. For the C2C variant, note that after reversing $(U,W)$ we may reverse the remaining edges (that have not been reversed yet) while the graph remains satisfied: after reversing $(U,W)$, $U$ is receiving $2N$ units of inflow and $W$ is receiving $N$ units, after reversing the remaining edges $U$ will be receiving $N$ units of inflow and $W$ will be receiving $2N$ units. 
\end{proof}

Bounded NCL (on graphs of treewidth $2$) is clearly in $\NP$, however, it is not clear whether there exists a pseudopolynomial algorithm for the problem, so we can not claim weak $\NP$-completeness. This is an open problem.

\begin{theorem}\label{thm:bncl-easy}
Bounded NCL (both C2E and C2C) is fixed parameter tractable with respect to parameter $tw+\Delta$.
\end{theorem}

\begin{proof}
The problems can be solved by dynamic programming on tree decompositions. For every bag, we consider all possible sequences of reversing the edges in that bag (since each edge may be reversed at most once, there are at most $(\Delta tw)!$ such sequences) that result in a valid solution for the subgraph induced by the subtree rooted at that bag. For C2C, we additionally require that a reversal sequence reverses exactly those edges that need to be reversed. For C2E, we additionally require that the reversal sequence reverses the target edge. We assume that the tree decomposition is given as a nice tree decomposition:

\begin{itemize}
\item \textbf{Leaf}: Enumerate all reversal sequences of incident edges (at most $\Delta!$). For the C2C variant, a sequence is valid if it reverses exactly the edges that must be reversed and additionally respects the minimum inflow of the vertex. For C2E, all sequences that respect the minimum inflow of the vertex are valid, unless the vertex is incident to the target edge in which case we additionally require that a valid reversal sequence reverses the target edge.

\item \textbf{Introduce}: Enumerate all reversal sequences of incident edges (at most $(\Delta tw)!$ sequences). A sequence is valid if it is valid for the vertex being introduced (as described in the leaf case) and also valid for the child bag (when restricted to the appropriate edges). 

\item \textbf{Forget}: A reversal sequence is valid if and only if it can be extended (by introducing reversal operations for the missing edges) to a valid reversal sequence for the child bag. Note that this is simply a projection of the reversal sequences valid for the child.

\item \textbf{Join}: A reversal sequence is valid if and only if it is valid for both child bags.
\end{itemize}

Since checking whether a reversal sequence is valid can be done in polynomial time (if the memoization tables are stored in a suitable data structure), the amount of work done in each bag is $O^*((\Delta tw)!)$. Since the number of bags can be assumed to be linear \cite{bodl97}, this gives a fixed-parameter tractable algorithm. 
\end{proof}

\section{Parametrization by Solution Length}

We consider parametrization by the length of the reconfiguration sequence $l$. This parameter does not bound the input size, but rather bounds the set of potential solutions considered. It asks not whether a reconfiguration sequence exists, but whether a reconfiguration sequence exists that has length (number of edge reversals) at most $l$. This parametrization does not apply to CGS (since it does not have such a notion), we thus consider this parametrization only for C2E and C2C NCL and their respective bounded variants.

\begin{theorem}\label{thm:lenhard}
(Unbounded) C2E, C2C and bounded C2E \textsc{Nondeterministic Constraint Logic} are $\W[1]$-hard with respect to $l$.
\end{theorem}

\begin{proof}
We show $\W[1]$-hardness by reduction from \textsc{$k$-Clique}, which is $\W[1]$-complete due to Downey and Fellows \cite{dfclique}:

Let $G=(V,E)$ be an instance of \textsc{$k$-Clique}. Let $n=|V|$ and let $\Delta$ be the maximum degree of $G$. We create a constraint graph with vertices $v_1,\ldots,v_n$ (copies of the original vertices from $G$) with minimum inflow equal to that vertex' degree in $G$, vertices $v_{(i,j)}$ with minimum inflow $2$ for all edges $(i,j)\in E$, and three additional vertices $U,V,W$ with minimum inflows $n-\Delta\cdot k$, $k(k-1)$ and $0$ respectively.

For every edge $(i,j)\in E$, two edges in the constraint graph are created: an edge $(v_i,v_{(i,j)})$ and an edge $(v_{(i,j)},v_j)$, both with weight $1$. Effectively, the edge $(i,j)$ is split and an additional vertex $v_{(i,j)}$ is added between vertices $v_i$ and $v_j$. In the initial configuration, both edges are oriented away from $v_{(i,j)}$ (and towards $v_i$ and $v_j$). Note that after adding these edges, the minimum inflow constraints in all vertices $v_1,\ldots,v_n$ are satisfied (since their minimum inflow is equal to their degree).

For every vertex $v_i$, add an edge $(U,v_i)$ to the graph with weight $\Delta$. The edge is oriented towards $U$ in the initial configuration, satisfying its inflow constraint.

For every edge $(i,j)\in E$, create an edge $(v_{(i,j)},V)$ of weight 2, oriented towards $v_{(i,j)}$ in the initial configuration. This satisfies the minimum inflow constraint of $v_{(i,j)}$.

Finally, create an edge $(V,W)$ of weight $k(k-1)$, oriented towards $V$ in the initial configuration, thus satisfying $V$'s minimum inflow. This edge is the target edge. Note that $W$'s minimum inflow is trivially satisfied.

The target edge $(V,W)$ can be reversed if and only if we have first reversed $\frac{k(k-1)}{2}$ of the weight $2$ edges to point in to $V$ (or else $V$'s inflow constraint would no longer be satisfied after reversing $(V,W)$. A weight $2$ edge $(v_{(i,j)},V)$ can be reversed if and only if we have first reversed the weight $1$ edges $(v_i,v_{(i,j)})$ and $(v_{(i,j)},v_j)$ (or else $v_{(i,j)}$'s minimum inflow would not remain satisfied). This in turn is possible if and only if both edges $(U,v_i)$ and $(U,v_j)$ have been reversed.

Because the minimum inflow of $U$ is $n-\Delta\cdot k$, we may reverse at most $k$ edges to point away from $U$. Since we may reverse the edge $(v_{(i,j)},V)$ if and only if both edges $(U,v_i)$ and $(U,v_j)$ have been reversed, the NCL instance thus constructed has a solution if and only if there are $k$ vertices in $G$ such that $G$ restricted to these vertices has $\frac{k(k-1)}{2}$ edges. This corresponds to these vertices forming a clique.

We have thus reduced \textsc{$k$-Clique} to C2E NCL. Note that if the instance created in the reduction has a solution of any length, then it has a solution of length at most $O(k^2)$ in which each edge reverses at most once. Hence the reduction is parameter-preserving, and works for both bounded and unbounded C2E NCL which we have thus shown $\W[1]$-hard.

To show $\W[1]$-hardness for C2C NCL, we replace the vertex $W$ with a latch gadget. The latch gadget has two states, and can change between these states if and only if the edge $(U,W)$ is pointing in to it. We ask whether there is a reconfiguration sequence that changes the state of the latch gadget, but where all the other edges are in their initial orientation. Such a sequence exists if and only if (in the original instance) it is possible to reverse the edge $(U,W)$, since we may first reverse $(U,W)$, activate the latch gadget, and then backtrack to return the rest of the graph to its original state. 
\end{proof}

We note that hardness is in the strong sense, since the weights of the edges and vertices are polynomial in the input size. However, the problems may become fixed parameter tractable if we include the maximum weight of an edge or vertex as parameter. This is an open problem.

The construction using the latch gadget does not work for bounded C2C NCL, since each edge can only be reversed once. Effectively, the hardness of NCL when parameterized by solution length is not due to the complexity of finding a reconfiguration sequence, but rather that of finding a suitable goal configuration. This is captured by the following theorem:

\begin{theorem}\label{thm:lenfpt}
Bounded C2C \textsc{Nondeterministic Constraint Logic} is fixed parameter tractable in $O^*(2^l)$ time.
\end{theorem}

\begin{proof}
If the number of edges that differ in orientation between the goal and start configuration exceeds $l$, we may reject immediately because is is impossible to reverse all of them in $l$ moves. We may thus assume at most $l$ edges need to be reversed. Since in bounded NCL each edge may be reversed at most once, feasible solutions consist of some order of reversing the edges that need to be reversed; reversing an edge that does not need to be reversed makes the instance unsolvable, and reversing an edge more than once is not possible. There are thus at most $l!$ solutions each of which can be evaluated in polynomial time. This can be improved to a $O^*(2^l)$ time algorithm using a dynamic programming approach, similar to the Held-Karp algorithm \cite{heldkarp} for TSP.
\end{proof}

\subparagraph*{Parametrization by Maximum Degree.} The graphs created by the reduction in the proof of Theorem \ref{thm:lenhard} may have arbitrary degree, while Hearn et al. \cite{hdbook} consider only graphs of maximum degree 3. We show that the hardness result of Theorem \ref{thm:lenhard} is strict in this regard, that is, when we parameterize by the maximum degree $\Delta$ in addition to $l$, the problem becomes fixed parameter tractable.

\begin{theorem}\label{thm:LDfpt}
Both bounded and unbounded C2E and C2C \textsc{Nondeterministic Constraint Logic} are fixed parameter tractable with respect to $l+\Delta$.
\end{theorem}

\begin{proof}
For the $C2E$ variants, note that an edge that is at a distance greater than or equal to $l$ from the target edge never needs to be reversed in a solution of length at most $l$: The only valid solution of length at most $1$ is to reverse the target edge, which indeed is at distance $0$ from itself. Consider a solution of length at most $l+1$ and that it reverses an edge $e$ at distance $l+1$ or greater from the target edge. After this edge is reversed, we are left with a solution of length at most $l$ which by induction only needs to reverse edges at distances less than $l$ from the target edge. We note that we can omit the reversal of $e$ from the reconfiguration sequence, since the vertex that gains inflow as a result of that reversal is not incident to any of the edges reversed subsequently, and hence omitting this reversal also results in a valid solution.

For any candidate solution, we only need to consider edges at distance at most $l-1$ from the target edge, of which there are at most $O((\Delta-1)^{l-1})$. We may omit the other edges from the graph (adjusting the minimum inflows of their incident vertices according to the initial orientations of the edges) and consider only the graph restricted to these edges and their incident vertices.

In the $C2C$ variants, similar to the proof of Theorem \ref{thm:lenfpt}, we may assume at most $l$ edges need to be reversed. We can thus bound the size of the set of edges that can be reversed in any candidate solution by $O(l (\Delta-1)^{l-1})$, since as before, for any of the $l$ edges needing to be reversed, we are only concerned with the edges at distance at most $l$ from that edge. As before, we can omit the remaining edges from the graph.

The problems thus admit kernelizations and therefore are fixed parameter tractable. 
\end{proof}

\section{Applications}\label{sec:applications}

As a result of our hardness proof for NCL, we (nearly) automatically obtain tighter bounds for other problems that reduce from NCL in their hardness proofs. If the reduction showing $\PSPACE$-hardness for a problem is bandwidth-preserving in the sense that the bandwidth of the resulting graph only depends on the bandwidth of the original constraint graph, then that problem remains $\PSPACE$-hard even when limited to instances of bounded bandwidth. Since reductions from NCL usually work by locally replacing AND and OR vertices in the graph with gadgets that simulate their functionality, such reductions are often bandwidth-preserving.

\decision{Independent Set Reconfiguration (IS-R)}{Graph $G=(V,E)$, independent sets $S_s,S_g\subseteq V$ of $G$, integer threshold $k$.}{Is there a sequence of independent sets $S_1,\ldots,S_n$, so that for all $i$, $S_{i+1}$ is obtained from $S_i$ by the addition or removal of one vertex, $|S_i|\geq k$ and $S_1=S_s,S_n=S_g$?}

This is the \emph{token addition-and-removal (TAR)} version of IS-R. It is also possible to define a \emph{token jumping (TJ)} variant (in which we obtain one independent set from the next by removing and immediately replacing a vertex) and a \emph{token sliding} variant (in which when we remove a vertex, we must replace it with an incident vertex, i.e. "sliding" the token/vertex along an edge).

\begin{theorem}\label{thm:isr}
TAR, TJ and TS versions of \textsc{\normalsize Independent Set Reconfiguration} are $\PSPACE$-complete on planar, maximum degree 3 graphs of bounded bandwidth.
\end{theorem}

\begin{figure}[b]
     \centering
     \hfill
     \subfloat[][OR vertex] {
         \includegraphics[scale=0.4]{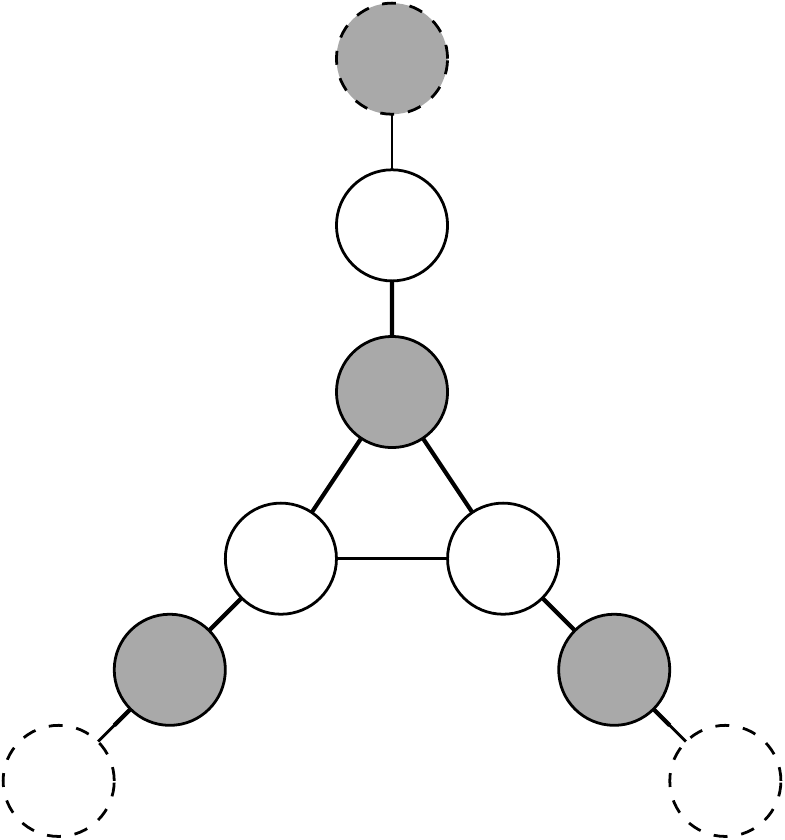}
         \label{fig:orvisr}
     }
     \hfill
     \subfloat[][AND vertex] {
         \includegraphics[scale=0.4]{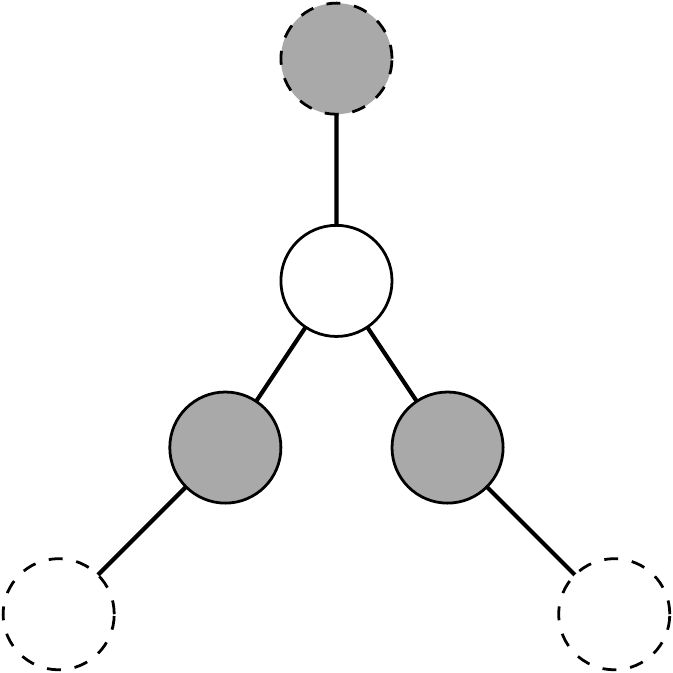}
         \label{fig:andvisr}
     }
     \hfill\null
     \caption{Gadgets used in the reduction from IS-R to NCL. Dashed vertices represent vertices that are part not of the gadget itself, but of other gadgets. The dark gray vertices represent an example independent set.}
     \label{fig:andorvisr}
\end{figure}

\begin{proof}
Hearn and Demaine \cite{hearn05} provide a reduction from IS-R to TJ and TS versions of IS-R. The gadgets used in their reduction are reproduced in figure \ref{fig:andorvisr}. Since the reduction works by replacing every vertex with a construction of at most 6 other vertices, the bandwidth of the graph created in the reduction is at most 6 times the bandwidth of the original constraint graph. The theorem thus follows immediately from our hardness result and Hearn and Demaine's reduction \cite{hearn05} and the fact that the TAR version is equivalent to TJ \cite{is-r}.
\end{proof}

Similarly to \textsc{\normalsize Independent Set Reconfiguration}, we can define reconfiguration versions of \textsc{\normalsize Vertex Cover} (VC-R), \textsc{\normalsize Feedback Vertex Set} (FVS-R), \textsc{\normalsize Induced Forest} (IF-R), \textsc{\normalsize Odd Cycle Transversal} (OCT-R) and \textsc{\normalsize Induced Bipartite Subgraph} (IBS-R). Note that for IS-R, IF-R, IBS-R the size of the vertex subset is never allowed to drop below the threshold, while for VC-R, FVS-R and OCT-R the size is never allowed to exceed the threshold. These problems are $\PSPACE$-complete on bounded bandwidth graphs by reduction from \textsc{\normalsize $H$-Word Reconfiguration} \cite{raman14}. We strengthen this result, noting that our proof follows the same reasoning as in \cite{raman14}:

\begin{theorem}
TAR, TJ and TS versions of \textsc{\normalsize Independent Set} (IS-R), \textsc{\normalsize Vertex Cover} (VC-R), \textsc{\normalsize Feedback Vertex Set} (FVS-R), \textsc{\normalsize Induced Forest} (IF-R), \textsc{\normalsize Odd Cycle Transversal} (OCT-R) and \textsc{\normalsize Induced Bipartite Subgraph} (IBS-R) are $\PSPACE$-complete on planar graphs of bounded bandwidth and low maximum degree.
\end{theorem}

\begin{proof}
By Theorem \ref{thm:isr}, IS-R is $\PSPACE$-complete on planar graphs of bounded bandwidth and maximum degree 3. Since an independent set is the complement of a vertex cover, the theorem holds for VC-R on maximum degree 3 graphs.

By replacing every edge in the graph by a triangle, we can reduce VC-R to FVS-R (since picking at least one vertex of every edge is equivalent to picking at least one vertex of every triangle), thus showing the theorem for FVS-R on maximum degree 6 graphs. We note that in such a graph a feedback vertex set is also an odd cycle transversal, thus also showing the theorem for OCT-R on maximum degree 6 graphs.

Finally, the theorem holds for IF-R and IBS-R on maximum degree 6 graphs by considering complements of solutions for FVS-R and OCT-R.
\end{proof}

Another related reconfiguration problem is that of \textsc{\normalsize Dominating Set Reconfiguration} (DS-R). In \cite{ds-r}, the authors show that DS-R is $\PSPACE$-complete on planar graphs of maximum degree six by reduction from NCL and $\PSPACE$-complete on graphs of bounded bandwidth by a reduction from VC-R. The following theorem that unifies these results follows immediately from our improved result concerning VC-R:

\begin{theorem}
The TAR version of \textsc{\normalsize Dominating Set Reconfiguration} is $\PSPACE$-complete, even on planar, maximum degree six graphs of bounded bandwidth.
\end{theorem} 

The puzzle game Rush Hour, in which the player moves cars horizontally and vertically with the goal to free a specific car from the board, is $\PSPACE$-complete \cite{bf02} when played on an $n\times n$ board. Another consequence of our result is that Rush Hour is $\PSPACE$-complete even on rectangular boards where one of the dimensions of the board is constant:

\begin{theorem}\label{thm:rushhour}
There exists a constant $k$, so that Rush Hour is $\PSPACE$-complete when played on boards of size $k\times n$.
\end{theorem}

\begin{proof}
Hearn and Demaine \cite{hearn05} provide a reduction from restricted NCL, showing how to construct AND and OR vertices as gadgets in Rush Hour and how to connect them together using straight line and turn gadgets. Given a planar constraint graph, it can be drawn in the grid (with vertices on points of the grid and edges running along the lines of the grid), after which vertices can be replaced by their appropriate gadgets and edges with the necessary line and turn gadgets. However, starting with a graph of bounded bandwidth does not immediately give a suitable drawing in a grid of which one of the dimensions is bounded (from which the theorem would follow, since all of the gadgets have constant size).

Given a restricted NCL graph of bounded bandwidth, we note that this graph also has bounded cutwidth \cite{bodl98}. Given such a graph with $n$ vertices and cutwidth at most $c$, it is possible to arrange it in a $(c+1)\times 3n$ grid, so that vertices of the graph are mapped to vertices of the grid, the edges of the graph run along edges of the grid, and no two edges or vertices get mapped to the same edge or vertex in the grid. This is achieved by placing all of the vertices of the graph in a single column (noting that 3 rows are required for each vertex since it has 3 incident edges) and (due to the graph having cutwidth $c$) the edges can be placed in the remaining $c$ columns (placing each edge in a distinct column). However, even when starting with a planar NCL graph, the resulting embedding may have crossings. These can be eliminated using the crossover gadget, noting that since the crossover gadget has constant size, the resulting graph can still be drawn in a $O(c)\times O(n)$ grid. We can now use the existing gadgets from \cite{hearn05} to finish the reduction, showing Rush Hour $\PSPACE$-complete on boards of which one of the dimensions is bounded by a constant. 
\end{proof}

Note that this result is likely to carry over to show other board games $\PSPACE$-complete on $n\times k$ boards, such as Sokoban or Plank Puzzles, which also reduce from NCL in their hardness proofs \cite{hdbook}.

This result is in contrast to \cite{pegs}, where it is shown that Peg Solitaire, when played on $k\times n$ boards, is linear time solvable for any fixed $k$. An important distinction here is that the length of a solution in a Peg Solitaire game is \emph{bounded} by the number of pegs (since every move removes one peg from play) whereas Rush Hour games are \emph{unbounded}: any length move sequence is possible, though obviously after an exponential number of moves, positions would be repeated. 

\section{Conclusions}

We have studied the parameterized complexity of constraint logic problems with regards to (combinations of) solution length, maximum degree and treewidth as parameters. As a main result, we showed that \textsc{\normalsize restricted Nondeterministic Constraint Logic} remains $\PSPACE$-complete on graphs of bounded bandwidth, strengthening Hearn and Demaine's framework \cite{hdbook}.

We showed that when parameterized by solution length, C2E, C2C and bounded C2E variants of \NCL are $\W[1]$-hard, but bounded C2C NCL becomes fixed parameter tractable. Essentially, the hardness of bounded C2E NCL when parameterized by solution length is due to the complexity of finding a suitable target configuration, rather than from that of determining a reconfiguration sequence. However, when parametrizing by maximum degree in addition to solution length, all cases become fixed parameter tractable.

When parameterized by treewidth, \textsc{Constraint Graph Satisfiability} becomes weakly $\NP$-complete (rather than strongly $\NP$-complete). When considering maximum degree in addition to treewidth, CGS becomes fixed parameter tractable. Note that this is in stark contrast to the complexity of NCL, which remains $\PSPACE$-complete even for fixed values of this parameter. This is an interesting example of how reconfiguration problems can be much harder than their decision variants.

Bounded NCL also becomes fixed parameter tractable when parameterized by treewidth and maximum degree. In the case where bounded NCL is parameterized by treewidth alone we proved weak $\NP$-hardness. It is an open problem whether there exists a pseudopolynomial algorithm for this case, or if the problem is also strongly $\NP$-complete.

By combining Wrochna's proof \cite{wrochna14} and Hearn and Demaine's \cite{hdbook} constraint logic techniques, we have managed to get the best of both worlds: for several reconfiguration problems, we showed hardness for graphs that are not only planar and have low maximum degree, but that also have bounded bandwidth. This not only strengthens Wrochna's results, but also makes it easier to prove hardness for reconfiguration on bounded bandwidth graphs by merging \textsc{\normalsize $H$-Word Reconfiguration} into the more convenient NCL framework.

Note that the constant $c$ in Theorem \ref{thm:c2cnclhard} has not been calculated precisely, but an informal analysis of Wrochna's proof \cite{wrochna14} and our reduction suggests it is very large (growing with the $12^\textrm{th}$ power of the original instance size). This raises two open questions: one is to determine a tighter bound on the value of $c$, and the other is to determine whether efficient algorithms exist for solving reconfiguration problems when the graph's bandwidth is bounded by more practical values. Some progress in this direction has already been made \cite{tjtrees}.

We have studied the parameterized complexity of CGS and NCL. Hearn and Demaine \cite{hdbook} have defined several other problems related to constraint graphs, including two player and multiple player constraint graph games, complete for classes such as $\EXPTIME$ and $\NEXPTIME$. Studying these problems in a parameterized setting gives rise to several interesting open problems.

\textbf{Acknowledgement.} I thank my advisor, Hans Bodlaender, for his guidance and useful comments and suggestions.

\bibliography{ref}{}

\end{document}